\documentclass[12pt]{article}
\usepackage[utf8]{inputenc}
\usepackage[margin=1in]{geometry} 
\usepackage{booktabs,caption,subcaption,dcolumn} 
\usepackage{amsfonts,bm}
\usepackage[flushleft]{threeparttable} 
\usepackage[table]{xcolor}
\usepackage{footmisc}
\usepackage{adjustbox}
\usepackage{multirow}
\usepackage{multicol}
\usepackage{enumerate} 

\usepackage{amsmath}
\usepackage{amsthm}
\usepackage{amssymb}
\theoremstyle{definition}
\newtheorem{assumption}{Assumption}
\newtheorem{remark}{Remark}
\theoremstyle{plain}
\newtheorem{theorem}{Theorem}
\newtheorem{lemma}{Lemma}

\usepackage{setspace}
\usepackage[style=apa]{biblatex}
\usepackage{hyperref}
\hypersetup{
    colorlinks=true,
    linkcolor=blue,
    urlcolor=magenta,
    citecolor=blue
}

\title{Conditional Triple Difference-in-Differences}
\author{Dor Leventer\thanks{Leventer: Tel-Aviv University, dorleventer@mail.tau.ac.il. For comments I thank Itay Saporta-Eksten, Oren Danieli, Zak Hirsch, Roee Levy, Michael Knaus. I gratefully acknowledge financial support from The Israel Pollak Fellowship Program for Excellence.}}
\date{\today}

\begin{document}

\maketitle

\begin{abstract}
Triple difference-in-differences designs are widely used to estimate causal effects in empirical work. Surveying the literature, we find that most applications include controls. We show that this standard practice is generally biased for the target causal estimand when covariate distributions differ across groups. To address this, we propose identifying a causal estimand by fixing the covariate distribution to that of one group. We then develop a double-robust estimator and illustrate its application in a canonical policy setting.
\end{abstract}

Triple difference-in-differences (TDID) designs are widely used to estimate causal effects, yet their theoretical foundations remain underdeveloped. In a survey of 66 well-cited TDID applications in top journals (see appendix for details), we find that over 70\% include control variables, typically in a linear regression with a three-way interaction. While motivated by standard econometrics textbooks \parencite[e.g.,][]{angrist2009mostly,cunningham2021causal}, this practice lacks formal justification. This paper provides the first theoretical analysis of TDID with controls, highlighting potential biases and proposing a theoretically grounded framework.

We begin by developing a TDID identification framework which incorporates controls for identification.\footnote{Adding controls for precision is discussed in the conclusion.} Our approach synthesizes two strands of the literature. The first, the unconditional TDID framework of \textcite{olden2022triple}, assumes that biases in unconditional parallel trends are invariant across two groups, denoted A and B. The second, the DID framework with controls of \textcite{callaway2021difference,sant2020doubly}, requires that parallel trends hold conditional on covariates. Combining the two, the new identification assumption is that biases in conditional parallel trends are invariant across two groups. We also consider TDID under two scenarios for the treatment assignment mechanism: one in which only eligible units in group A are treated, and another in which eligible units are treated in both groups.

We derive two main identification results. First, we show that the standard approach, estimating a DID with controls separately for each group and differencing the results, fails to recover the target causal estimand when covariate distributions differ across groups. The target estimand, either the average treatment effect on the treated (ATT) for group A or the difference in ATTs between groups A and B, is determined by the treatment assignment mechanism. The intuition is straightforward: each group’s ATT equals its DID plus the average bias in conditional parallel trends. If covariates are distributed differently, the plausible case in most applications, the averaged bias terms do not cancel out, and the (triple) difference in DID no longer identifies the causal estimand.

Second, we show that an alternative approach, taking the difference in conditional DIDs under group A’s covariate distribution, identifies a causal estimand. Depending on the treatment assignment, this target is either the ATT of group A or the expected difference in conditional ATTs between groups A and B, evaluated at group A’s covariate distribution. While the estimand can be viewed as a practical adjustment—identifying what is feasible, it can also be motivated economically, as differences in treatment effects holding the covariate distribution constant may be the estimand of interest in many applications.

For estimation, we develop a double-robust (DR) procedure that extends both parametric and semi-parametric DID methods to the TDID-with-controls setting \parencite{sant2020doubly, chang2020double}. The estimator reweights group B’s DID using the covariate distribution of group A, following the identification result. 
We illustrate this framework using the canonical minimum wage study by \textcite{card1993minimum}. 

To fix ideas, consider the following examples. A first example is \textcite{muralidharan2017cycling}, who study the impact of a bicycle program on girls’ 9th grade enrollment by comparing younger and older cohorts (first difference), boys and girls (second difference), and across two Indian states (third difference). The policy was implemented in only one state and targeted only girls. In our framework, TDID assumes that the bias in parallel trends between genders across cohorts, under the counterfactual without the policy, is equal across states conditional on covariates, such as caste, religion, or access to infrastructure. The standard approach computes a DID within each state and differences the two. This equals the ATT plus a bias term capturing differences in average conditional-trend bias across states. This bias is non-zero when covariate distributions differ across the two states. However, reweighting the control state’s DID to match the covariate distribution in the treated state identifies the ATT. 

Second, \textcite{kleven2019children} examine the impact of parenthood on earnings by comparing individuals before and after childbirth (first difference), early and late age at first birth (second difference), and across mothers and fathers (third difference). Both genders become parents after their first child is born, making both groups treated in the post-childbirth period. A TDID design assumes that the bias in parallel trends in earnings across age-at-first-birth groups, under the counterfactual without children, is equal across genders conditional on covariates, such as education. The standard approach computes a DID separately for mothers and fathers and differences the two, yielding the difference in ATTs between mothers and fathers plus a bias term reflecting differences in average conditional-trend bias between genders.  However, reweighting the DID for fathers to match the covariate distribution of mothers identifies the difference in conditional ATTs, averaged over the mothers' covariate distribution. This estimand is policy-relevant, as it isolates gender differences in the impact of parenthood while holding covariate distributions fixed.

This paper contributes to the econometric literature on TDID and DID. We build on \textcite{olden2022triple}, who formalize TDID as a single bias-in-parallel-trends assumption, and extend recent work on conditional parallel trends in DID settings \parencite{zimmert2018efficient, chang2020double, sant2020doubly}. The paper makes four contributions. First, we lay the theoretical foundations for identification in TDID designs with covariates. Second, we show that the standard TDID empirical strategy is biased when covariate distributions differ across groups. Third, we develop a framework that identifies a causal target estimand, and we clarify this estimand is determined by the treatment assignment mechanism. Lastly, we develop estimators for the identifiable estimand. An accompanying R package for implementation code is available at \href{https://github.com/dorleventer/tdid}{\texttt{tdid}}.

The remainder of the paper is structured as follows. Section \ref{sec:setup} introduces the notation and basic setup. Section \ref{sec:identification} presents the identification framework and results. Section \ref{sec:est} discusses the estimation framework. Section \ref{sec:emp_app} illustrates the results in an empirical application. Section \ref{sec:conlusion} concludes.

\section{Notation and Setup}\label{sec:setup}

We consider a 2-by-2-by-2 TDID setup, where outcomes are observed across two time periods, for units in two treatment cohorts and two distinct groups. TDID designs differ in their treatment assignment mechanisms, which in turn determine the identified causal estimand (discussed below in Section \ref{sec:identification}). Returning to the examples above, in \textcite{muralidharan2017cycling} only units in one group were treated, while in \textcite{kleven2019children} units from both groups were treated.

We now formalize a general framework that captures both cases by explicitly distinguishing between treatment assignment mechanisms. Let $i \in \{1,\dots,n\}$ index individuals observed in periods $t \in \{1,2\}$. To allow for the possibility that only certain units are eligible for treatment, we separate eligibility from treatment. Let $E_i \in \{2, \infty\}$ denote the period in which unit $i$ becomes eligible for treatment, with $E_i = \infty$ for those never eligible. Let $W_{i,t} \in \{0,1\}$ denote treatment status at time $t$, and let $G_i \in \{a,b\}$ denote group membership (e.g., girls and boys, or mothers and fathers). We assume no units are treated in period 1, i.e., $W_{i,1} = 0$ for all $i$. Treatment assignment in period 2 is governed by the following:

\begin{assumption}\label{A.treatment_ass}
Either (i) $W_{i,2} = 1_{{E_i = 2, G_i = a}}$, or (ii) $W_{i,2} = 1_{{E_i = 2}}$.
\end{assumption}

\noindent
Assumption \ref{A.treatment_ass} generalizes treatment assignment structures in TDID designs. The policy in \textcite{muralidharan2017cycling} corresponds to case (i); \textcite{kleven2019children} corresponds to case (ii).

Under the stable unit treatment value assumption (SUTVA) \parencite{rubin1980randomization}, we denote the potential outcome when unit $i$ remains untreated in both periods as $Y_{i,t}(0)$, and when not treated in the first period and treated in the second period as $Y_{i,t}(1)$. Observed and potential outcomes are linked by the consistency assumption, $Y_{i,t} = Y_{i,t}(W_{i,2})$.\footnote{This potential outcome notation follows the 2-by-2 potential outcome notation in \textcite{roth2023s} for DID.} Finally, let $X_i$ denote a set of time-invariant covariates with support $\mathcal{X}$.

\section{Identification}\label{sec:identification}

This section starts with articulating several causal estimands and formulated the TDID with controls framework is formulated. Finally, the two identification results are presented. All proofs are given in Appendix \ref{sec:appendix_proofs}. An analytical example, which numerically shows the results, is presented in Appendix \ref{sec:example}. 

\subsection{Estimands}
Let 
$$ATT(g)=\mathbb{E}[Y_{2}(1)-Y_{2}(0)\mid G=g,E=2]$$
be the ATT of group $g$ at the second time period, and  
$$CATT(g,x)=\mathbb{E}\left[Y_{2}\left(1\right)-Y_{2}\left(0\right)\mid G=g,E=2,X=x\right]$$
be the ATT conditional (CATT) on covariate value $X=x$ of group $g$ at the second time period.
As we show below, the identified causal estimand depends on Assumption \ref{A.treatment_ass}. Under Assumption \ref{A.treatment_ass}(i), the target estimand is $ATT(a)$; under Assumption \ref{A.treatment_ass}(ii), one can target either $ATT(a) - ATT(b)$ or
$$\mathbb{E}[CATT\left(a,X\right)-CATT\left(b,X\right)\mid G=a,E=2].$$
This last estimand captures the difference in conditional treatment effects between groups $a$ and $b$, averaged over the covariate distribution of group $a$. By holding the covariate distribution fixed, it isolates differences in treatment effects.
Returning to the example of \textcite{kleven2019children}, this estimand corresponds to the difference in the conditional effect of parenthood between mothers and fathers, evaluated at the covariate distribution of mothers. In this application, the last estimand may be of particular interest, as it isolates gender-based differences in the effect of parenthood.

\subsection{Identification Assumptions}

The first identification assumption, commonly referred to as no anticipation \parencite{callaway2021difference}, equates the expected value of the untreated and treated potential outcomes in the pre-treatment period. 
\begin{assumption}[Conditional No Anticipation]\label{A.no.ant.1}
For group $g\in\{a,b\}$ and all $x\in\mathcal{X}$
$$\mathbb{E}\left[Y_{1}(0)\mid G=g,E=2,X=x\right]=\mathbb{E}\left[Y_{1}\mid G=g,E=2,X=x\right]$$
\end{assumption}

\noindent
The second component concerns parallel trends. Standard DID with controls assumes a conditional parallel trends condition: the trend in untreated potential outcomes is equal between treated and never-treated units, conditional on covariates. Define the bias in conditional parallel trends for group $g$ as
\begin{align*}
\delta_{\mathrm{CPT}}\left(g,x\right) & =\mathbb{E}\left[Y_{2}\left(0\right)-Y_{1}\left(0\right)\mid G=g,E=2,X=x\right]\\
 & -\mathbb{E}\left[Y_{2}\left(0\right)-Y_{1}\left(0\right)\mid G=g,E=\infty,X=x\right].
\end{align*}
The conditional parallel trends assumption of \textcite{callaway2021difference} corresponds to setting $\delta_{\mathrm{CPT}}(g, x) = 0$ for $g \in \{a, b\}$. In contrast, the TDID framework with covariates allows for non-zero $\delta_{\mathrm{CPT}}(g, x)$, but assumes that the bias is equal across groups.
This is analogous to \textcite{olden2022triple}, who assume equal bias in unconditional parallel trends.
\begin{assumption}[Bias in Conditional Parallel Trends is Group-Invariant]\label{A.same_bias}
For all $x\in\mathcal{X}$, $\delta_\mathrm{CPT}(a,x)=\delta_\mathrm{CPT}(b,x)$.
\end{assumption}
\noindent
To compare or reweight conditional expectations across groups, we also require an assumption often termed overlap.
Let 
$p\left(g,e,x\right)=P\left(\left(G,E\right)=\left(g,e\right)\mid X=x\right)$
denote the probability of belonging to group $g$ with eligibility status $e$ given covariates $X = x$. (\textcite{callaway2021difference} term this a "generalizied" propensity score). We assume the following:
\begin{assumption}[Overlap]\label{A.overlap} For all $x\in\mathcal{X}$, $$p\left(a,2,x\right),p\left(a,\infty,x\right),p\left(b,2,x\right),p\left(b,\infty,x\right)\in(0,1).$$
\end{assumption}

Assumption \ref{A.same_bias} is the key identifying assumption. Since treatment is not randomly assigned, it should be justified using contextual knowledge and economic modeling.
To guide the assessment of its plausibility, we suggest considering two components. First, whether the two groups (e.g., mothers and fathers) select into treatment (e.g., when to have children) based on similar factors. Second, whether, conditional on these factors, the bias in untreated potential outcomes would evolve similarly across groups.
Finally, because the assumption is conditional on covariates, the choice of covariates is critical. The assumption will only hold if $X$ captures the relevant dimensions of selection. Justifying it therefore requires careful consideration of the relationship between the covariates and the institutional setting.

\subsection{Identification Results}

Consider the following empirical strategy: run a DID regression with controls separately for each group, then take the difference. This corresponds to the typical TDID empirical strategy \parencite[see][or the literature survey]{angrist2009mostly,cunningham2021causal}. The result below shows that, without further assumptions, the descriptive estimand\footnote{We differentiate between descriptive estimands, which are functions of observed data, and causal estimands, which depend on both observed data and potential outcomes, following \textcite{abadie2020sampling}.} corresponding to this strategy does not identify the causal parameter of interest. Let
\begin{align*}
D\left(g,x\right) & =\mathbb{E}\left[Y_{2}-Y_{1}\mid G=g,E=2,X=x\right]\\
 & -\mathbb{E}\left[Y_{2}-Y_{1}\mid G=g,E=\infty,X=x\right]
\end{align*}
denote the DID descriptive estimand for group $g\in\{a,b\}$, conditional on covariate 
value $X=x$. We now turn to the first main result. 
\begin{theorem}\label{theorem.1}
Assume Assumption \ref{A.no.ant.1} holds for both $g\in\{a,b\}$, and Assumptions \ref{A.same_bias} and \ref{A.overlap} also hold. Let $\delta_{\mathrm{CPT}}(x)$ be the (group-invariant) bias in conditional parallel trends. Then
\begin{equation*}
    \mathbb{E}[D(a,X)\mid G=a,E=2]-\mathbb{E}[D(b,X)\mid G=b,E=2]=\tau_{\mathrm{T}1}+\mathrm{bias}
\end{equation*}
where 
$$\tau_{\mathrm{T}1}=\begin{cases}
ATT\left(a\right) & \text{under Assumption \ref{A.treatment_ass}(i),}\\
ATT\left(a\right)-ATT\left(b\right) & \text{under Assumption \ref{A.treatment_ass}(ii),}
\end{cases}$$
and 
$$\mathrm{bias}=\mathbb{E}\left[\delta_{\mathrm{CPT}}\left(X\right)\mid G=a,E=2\right]-\mathbb{E}\left[\delta_{\mathrm{CPT}}\left(X\right)\mid G=b,E=2\right].$$
\end{theorem}

\noindent
Theorem \ref{theorem.1} shows that the typical TDID regression estimand, the difference between groups of the conditional DIDs averaged over the covariate distribution of its respective group (the left-hand side expression) is equal to a causal estimand plus a bias term. The causal estimand depends on the treatment assignment mechanism: it is $ATT(a)$ under Assumption \ref{A.treatment_ass}(i) or $ATT(a)-ATT(b)$ under Assumption \ref{A.treatment_ass}(ii). The bias term is equal to the difference between the two groups in the average bias calculated of conditional parallel trends within each group.  

The implication is that, under the TDID design, simply differencing DIDs across groups does not identify a causal estimand unless an additional condition holds. Specifically, identification would require that the average bias in conditional parallel trends is equal across groups. This condition holds if the covariate distributions of groups A and B are identical. However, the rationale for Assumption \ref{A.same_bias} is precisely that biases in unconditional parallel trends differ, hence the need to condition on covariates. As a result, if Assumption \ref{A.same_bias} is justified, the bias term in Theorem \ref{theorem.1} is unlikely to be zero, and the causal estimand is not identified by the difference in group-specific DIDs within the conditional TDID framework.

The previous result was negative: it showed that the standard empirical strategy in TDID does not identify a causal estimand. The next result is positive, showing that a well-defined causal estimand is identifiable within the conditional TDID framework.
\begin{theorem}\label{theorem.2}
Assume the conditions in Theorem \ref{theorem.1} hold. Then 
\begin{equation*}
\tau_{\mathrm{T}2} = \mathbb{E}[D(a,X)\mid G=a, E=2]-\mathbb{E}[D(b,X)\mid G=a, E=2],
\end{equation*}
where 
\begin{equation*}
\tau_{\mathrm{T}2}=\begin{cases}
ATT\left(a\right) & \text{under Assumption \ref{A.treatment_ass}(i),}\\
\mathbb{E}[CATT\left(a,X\right)-CATT\left(b,X\right)\mid G=a,E=2]  & \text{under Assumption \ref{A.treatment_ass}(ii).}
\end{cases}
\end{equation*}
\end{theorem}

\noindent
Theorem \ref{theorem.2} shows that a causal estimand is identified by taking the difference between the DID with controls for group A and a reweighted DID for group B, where group B is weighted to match the covariate distribution of group A. This result requires no additional assumptions beyond those already stated.

\section{Estimation}\label{sec:est}

The identification result in Theorem \ref{theorem.2} motivates estimation of the two expectations $\mathbb{E}[D(a,X)\mid G=a, E=2]$ and $\mathbb{E}[D(b,X)\mid G=a, E=2]$. We begin by establishing equality between these terms and doubly robust (DR) descriptive estimands, building on \textcite{sant2020doubly}. We then discuss an estimation procedure and inference. A simulation study is provided in Appendix \ref{simulation}.

First, we define weights and outcome regressions. 
\begin{align*}w_{\mathrm{T}}\left(g,e\right) & =\frac{1_{\left\{ G_{i}=g,E_{i}=e\right\} }}{\mathbb{E}\left[1_{\left\{ G=g,E=e\right\} }\right]},\\
w_{\mathrm{C}}\left(g,g^{\prime},e,e^{\prime},x\right)= & \frac{1_{\left\{ G_{i}=g^{\prime},E_{i}=e^{\prime}\right\} }}{\mathbb{E}\left[1_{\left\{ G=g,E=e\right\} }\right]}\frac{p\left(g,e,x\right)}{p\left(g^{\prime},e^{\prime},x\right)},\\
m\left(g,e,x\right)= & \mathbb{E}\left[Y_2-Y_1\mid G=g,E=e,X=x\right].
\end{align*}
Next, we define score functions for outcome regression (OR), inverse probability weighted (IPW) and DR.
\begin{align*}
\psi_{OR}\left(g;O_{i},\theta\right) & =w_{\mathrm{T}}\left(g,2\right)\left(Y_{i,2}-Y_{i,1}\right)-w_{\mathrm{T}}\left(g,2\right)m\left(g,\infty,X_{i}\right),\\
\psi_{IPW}\left(g;O_{i},\theta\right) & =\left(w_{\mathrm{C}}\left(g,g,2,2,X_{i}\right)-w_{\mathrm{C}}\left(g,g,2,\infty,X_{i}\right)\right)\left(Y_{i,2}-Y_{i,1}\right),\\
\psi_{DR}\left(g;O_{i},\theta\right) & =\psi_{IPW}\left(g;O_{i},\theta\right)+\left(w_{\mathrm{T}}\left(g,2\right)-w_{\mathrm{C}}\left(g,g,2,2,X_{i}\right)\right)m\left(g,2,X_{i}\right)\\
 & -\left(w_{\mathrm{T}}\left(g,2\right)-w_{\mathrm{C}}\left(g,g,2,\infty,X_{i}\right)\right)m\left(g,\infty,X_{i}\right),
\end{align*}
where $O_i=(Y_{i,2},Y_{i,1},E_i,G_i,X_i)$ is the observed data and $\theta$ are the nuisance parameters, which in this context are the four generalized propensity scores, $p\left(a,2,x\right),p\left(a,\infty,x\right),p\left(b,2,x\right)$ and $p\left(b,\infty,x\right)$, and three outcome regressions, $m\left(a,\infty,x\right),m\left(b,2,x\right)$ and $m\left(b,\infty,x\right)$. From \textcite{sant2020doubly,callaway2021difference} or slight configuration of the lemma below, it follows that under the TDID design (abusing notation by omitting $O_i$ and $\theta$ for compactness)
\begin{equation}\label{eq:iden_did}
    \mathbb{E}\left[D\left(a,X\right)\mid G=a,E=2\right]=\mathbb{E}\left[\psi_{OR}\left(a\right)\right]=\mathbb{E}\left[\psi_{IPW}\left(a\right)\right]=\mathbb{E}\left[\psi_{DR}\left(a\right)\right].
\end{equation}
After establishing identification of the first expectation using a DR approach, we now turn to the second term, $\mathbb{E}[D(b,X)\mid G=a,E=2]$. To differentiate between the score functions, we add a W (short for weighted) to the notation. Let 
\begin{align*}
\psi_{WOR}\left(O_{i},\theta\right) & =w_{\mathrm{T}}\left(a,2\right)m\left(b,2,X_{i}\right)-w_{\mathrm{T}}\left(a,2\right)m\left(b,\infty,X_{i}\right),\\
\psi_{WIPW}\left(O_{i},\theta\right) & =\left(w_{\mathrm{C}}\left(a,b,2,2,X_{i}\right)-w_{\mathrm{C}}\left(a,b,2,\infty,X_{i}\right)\right)\left(Y_{i,2}-Y_{i,1}\right),\\
\psi_{WDR}\left(O_{i},\theta\right) & =\psi_{WIPW}\left(O_{i},\theta\right)+\left(w_{\mathrm{T}}\left(a,2\right)-w_{\mathrm{C}}\left(a,b,2,2,X_{i}\right)\right)m\left(b,2,X_{i}\right)\\
 & -\left(w_{\mathrm{T}}\left(a,2\right)-w_{\mathrm{C}}\left(a,b,2,\infty,X_{i}\right)\right)m\left(b,\infty,X_{i}\right).
\end{align*}
We formally establish identification in the following lemma (where we again abuse notation by omitting $O_i$ and $\theta$). 

\begin{lemma}\label{lemma.wdr}
Assume the conditions in Theorem \ref{theorem.1} hold. Then 
\begin{align*}
\mathbb{E}\left[D\left(b,X\right)\mid G=a,E=2\right]&=\mathbb{E}\left[\psi_{WOR}\right]=\mathbb{E}\left[\psi_{WIPW}\right]=\mathbb{E}\left[\psi_{WDR}\right].
\end{align*}
\end{lemma}
\noindent The proof is given in Appendix \ref{sec:appendix_proofs}, and loosely follows the structure of the proof in \textcite{callaway2021difference}. To our knowledge, this is a novel result establishing double-robust identification for a reweighted DID estimand in a TDID setting. It enables identification of TDID estimands under group-specific reweighting and is key to estimating the causal effect defined in Theorem \ref{theorem.2}. Specifically, combining \eqref{eq:iden_did}, Lemma \ref{lemma.wdr}, and Theorem \ref{theorem.2} implies that $\tau_{\mathrm{T}2}=\mathbb{E}\big[\psi_{\mathrm{DR}}(a)\big]-\mathbb{E}\big[\psi_{\mathrm{WDR}}\big]$.

Let $\widehat{\tau}_{\mathrm{T}2}$ be an estimator for $\tau_{\mathrm{T}2}$, for which we discuss a procedure below. Inference hinges on a central limit theorem $\sqrt{n}(\widehat{\tau}_{\mathrm{T}2}-\tau_{\mathrm{T}2})\rightarrow N(0,V)$. 
This can be achieved by two different approaches on estimated nuisance parameters, denoted $\widehat{\theta}$. First, as discussed in \textcite{sant2020doubly}, one can use parametric models. If these models are correctly specified and regularity conditions hold asymptotic normality of $\widehat{\tau}_{\mathrm{T}2}$ follows. For this approach, using ordinary least squares (OLS) to model the outcome regression ($m(.)$) and logistic regression to model the propensity score ($p(.)$) are straightforward choices. Second, as discussed in \textcite{chang2020double}, one can implement non-parametric models in a de-biased machine learning framework \parencite{chernozhukov2018double}. This requires algorithms with sufficient convergence rates and the use of cross fitting.

For a given choice of models to estimate the nuisance parameters, estimation can be carried out as follows. First, estimate $\widehat{\theta}$. Second, estimate $\widehat{\psi}_{\mathrm{DR}}(a;O_i,\widehat{\theta})$ and $\widehat{\psi}_{\mathrm{WDR}}(O_i,\widehat{\theta})$, which are estimators of $\psi_{\mathrm{DR}}(a;O_i,\theta)$ and $\psi_{\mathrm{WDR}}(O_i,\theta)$, respectively, constructed by substituting parameters in $\theta$ with fitted models in $\widehat{\theta}$ and expectations with sample means. Third, calculate
$\widehat{\tau}_{\mathrm{T}2}=\mathbb{E}_n\big[\widehat{\psi}_{\mathrm{DR}}(a;O_i,\widehat{\theta})-\widehat{\psi}_{\mathrm{WDR}}(O_i,\widehat{\theta})\big]$, where $\mathbb{E}_n[.]$ denotes the sample mean.  Finally, estimate the asymptotic variance of $\widehat{\tau}_{\mathrm{T}2}$ using $\widehat{V}=\frac{1}{n}\sum_i\eta \left(O_{i},\widehat{\theta},\widehat{\tau}_{\mathrm{T}2}\right)^2$, where $$\eta(O_{i},\theta,\tau) = \psi_{\mathrm{DR}}(a;O_i,\theta)-\psi_{\mathrm{WDR}}(O_i,\theta) - w_{\mathrm{T}}(a,2)\tau_{\mathrm{T}2},$$
and calculate the standard errors of the estimator using $\sqrt{\widehat{V}/n}$.

\section{Empirical Application}\label{sec:emp_app}

We illustrate our framework using the canonical minimum wage study by \textcite{card1993minimum}.\footnote{This application is also discussed in \textcite[chapter 9.5]{cunningham2021causal} as a TDID example.} In November 1992, New Jersey raised its minimum wage from \$4.25 to \$5.05, while Pennsylvania’s remained unchanged. The original analysis used a two-by-two DID design comparing employment in restaurants before and after the reform across the two states.

We revisit the study to ask a different question: did the reform have a different impact on low- vs. high-wage restaurants? We illustrate how this question can be addressed using the TDID design. Since both groups are treated if located in New Jersey in the post-reform period, this aligns with Assumption~\ref{A.treatment_ass}(ii). Suppose that average potential employment under the lower minimum wage would have followed different trends in New Jersey and Pennsylvania—that is, standard DID fails due to bias in parallel trends. The TDID identification assumption is that this bias is constant across wage groups, conditional on covariates $X$ (Assumption~\ref{A.same_bias}). Because wage groups may differ systematically, this conditional version is more plausible than an unconditional assumption. The resulting estimand is the difference in conditional ATTs, averaged over the covariate distribution of low-wage restaurants ($\tau_{\mathrm{T}2}$ in Theorem~\ref{theorem.2}), which directly adresses the research question.

Let $Y_{i,t}$ denote employment at restaurant $i$ in period $t \in \{\text{pre}, \text{post}\}$, and let $E_i \in \{\text{NJ}, \text{PA}\}$ indicate the state. Restaurants are grouped based on their pre-treatment average starting wage: group A includes those at or below the median (\$4.50), and group B includes those above it. Let $G_i \in \{a, b\}$ denote this group indicator. Treatment assignment follows $W_{i,\text{pre}} = 0$ and $W_{i,\text{post}} = 1_{\{E_i = \text{NJ}\}}$. Covariates $X_i$ include soda price, number of managers, and opening hour, measured in the pre-treatment period. After dropping observations with missing values, the sample size is 695.

We implement the following estimators. First, we estimate DID and TDID using standard linear regressions, both with and without controls, based on two-way interactions for DID and three-way interactions for TDID.\footnote{For completeness, we write the regressions explicitly. The (group-specific) DID regression is 
$$g\in\{a,b\},\quad Y_{i,t}=\alpha_{0}^g+\alpha_{1}^g1_{\{E_{i}=\text{NJ}\}}+\alpha_{2}^g1_{\{t=\text{post}\}}+\alpha_{3}^g1_{\{E_{i}=\text{NJ},t=\text{post}\}}+\phi^gX_{i}+u_{i}$$
with target coefficients $\alpha^a_3$ and $\alpha^b_3$, and the TDID regression is 
\begin{align*}
Y_{i,t} & =\beta_{0}+\beta_{1}1_{\{E_{i}=\text{NJ}\}}+\beta_{2}1_{\{t=\text{post}\}}+\beta_{3}1_{\left\{ G_{i}=a\right\} }\\
 & +\beta_{4}1_{\{E_{i}=\text{NJ},t=\text{post}\}}+\beta_{5}1_{\{E_{i}=\text{NJ},G_{i}=a\}}+\beta_{6}1_{\left\{ t=\text{post},G_{i}=a\right\} }\\
 & +\beta_{7}1_{\{E_{i}=\text{NJ},t=\text{post},G_{i}=a\}}+\varphi X_{i}+\varepsilon_{i}
\end{align*}
with target cofficient $\beta_7$. A regression "without controls" means omitting $X_i$ from the above equations.  
} Controls are added linearly with no interactions. Second, we apply an OR approach, since it transparently illustrates covariate reweighting.
Specifically, we estimate separate OLS outcome models of $Y$ on covariates $X$ for each combination of time period ($t$), state ($E$), and group ($G$), yielding eight models in total. To compute the DID for group A, we use the four models trained on group A to predict outcomes for treated units of group A under all $(t, W)$ combinations, calculate unit-level treatment effects, and average them. This corresponds to an estimator of $\mathbb{E}\big[\psi_{OR}(a)\big]$ (Section \ref{sec:est}). We repeat the same procedure for group B and estimate $\mathbb{E}\big[\psi_{OR}(b)\big]$. To estimate the weighted DID for group B under the covariate distribution of group A, we predict outcomes for treated units from group A using the models trained on group B, and calculate the DID using these predictions. This corresponds to an estimator of $\mathbb{E}\big[\psi_{WOR}\big]$.

The results, presented in Table~\ref{tab:tdid_results}, illustrate differences between conventional estimators and those motivated by the identification framework. In the OLS specifications, the triple difference (column 4) is modest, 1.35 without controls and 1.56 with controls. The OR estimator without controls yields similar conclusions. When controls are included, differencing the DID of group A (column 1) from that of group B (column 2) produces a larger estimate of 5.42. However, this estimator does not recover a valid causal estimand (Theorem~\ref{theorem.1}).
In contrast, differencing the DID of group A from the weighted DID of group B (column 3), which applies group A’s covariates to group B’s outcome model, yields an estimate of 4.23 (column 5). This estimator is unbiased (identification-wise) for the average difference in conditional ATTs under the covariate distribution of group A (Theorem~\ref{theorem.2}). However, none of the estimates are statistically significant, likely due to the small sample size, and should be viewed as illustrative rather than conclusive.

\begin{table}[t]
\centering
\caption{Application Results}
\label{tab:tdid_results}
\begin{tabular}[t]{lccccc}
\toprule
Method & DID A & DID B & W DID B & A - B & A - WB\\
& (1) & (2) & (3) & (4) & (5) \\
\midrule
OLS no Controls & 2.84 & 1.49 &  & 1.35 & \\
 & (2.38) & (2.67) &  & (3.57) & \\
OLS w/ Controls & 3.08 & 1.52 &  & 1.56 & \\
 & (2.04) & (2.27) &  & (3.05) & \\
OR no Controls & 2.72 & 1.47 &  & 1.25 & \\
 & (2.60) & (2.76) &  & (3.66) & \\
OR w/ Controls & 3.76 & -1.67 & -0.47 & 5.42 & 4.23\\
 & (3.37) & (4.40) & (3.52) & (5.51) & (4.86)\\
\bottomrule
\end{tabular}
\caption*{
\footnotesize \textit{Notes:} The table presents estimates on the effect of minimum wage reform on  employment, using replication data from \textcite{card1993minimum}. The outcome variable is full-time employment. Group A are restaurants with starting wages smaller than or equal to median starting wages in the pre-treatment period, and group B are restaurants with starting wages greater than the median starting wages in the pre-treatment period. "DID A" and "DID B" refer to DID estimates of group A and B, respectively. "W DID B" refers to weighted DID estimates of group B according to covariate values of group A. "A - B" refers to the difference between DID A and DID B. "A - WB" refers to the difference between DID A and W DID B. "OLS" refers to a conventional linear regression specification, and "OR" refers to an outcome regression method using linear regression. "no Controls" and "w/ Controls" refers to a specification without and with controls, respectively. Standard errors are reported in parentheses. For OR specifications, bootstrap standard errors are reported.
}
\end{table}

\section{Conclusion}\label{sec:conlusion}

This paper provides a theoretical analysis of the triple difference (TDID) design with time-invariant covariates. Our main contributions are as follows. First, we lay the theoretical foundations for identification in TDID designs when conditioning on covariates. Second, we show that the standard TDID empirical strategy can be biased if covariate distributions differ across groups. Third, we show that re-weighed by the covariates of one group identifies a well-defined causal estimand. Lastly, we construct double-robust estimators that recover this estimand under standard assumptions.

While our setup focuses on a simple 2-by-2-by-2 panel structure, several natural extensions are possible. One is to allow for multiple time periods, in the spirit of \textcite{callaway2021difference}. Another is to adapt the framework to repeated cross-section data. Additional extensions could draw on recent developments in the DID literature, including time-varying covariates \parencite{caetano2022difference}, continuous treatments \parencite{callaway2024difference}, and settings with compositional changes \parencite{sant2023difference}.
Our results underscore that TDID presents distinct identification challenges relative to standard DID, motivating further research on TDID designs.

We finish with several remarks of practical relevance. 

\begin{remark}[Differencing two DR estimators]
Theorem \ref{theorem.1} shows that calculating a separate DR estimator for group A à la \textcite{sant2020doubly}, and then a separate DR estimator for group B and differencing them, is biased for the target causal estimand if the covariate distributions differ between groups. 
\end{remark}

\begin{remark}[Testing for bias]
Under Assumption~\ref{A.treatment_ass}(i), the bias in conditional parallel trends ($\delta_{\mathrm{CPT}}(X)$ in Theorem~\ref{theorem.1}) is identifiable (see Lemma~\ref{lemma.catt_diff} in the Appendix). This allows researchers to test whether the bias is non-zero in conventional TDID estimators. Moreover, if the bias is found to be significant, one can apply a bias-correction procedure to recover the ATT.
\footnote{Lemma~\ref{lemma.catt_diff} shows that the DID for group B identifies the (group-invariant) bias in conditional parallel trends. Therefore, one can model the DID for group B, predict its mean under each group's covariate distribution, and take their difference to quantify the bias in Theorem~\ref{theorem.1}. Subtracting this quantity yields a bias-corrected estimator of the ATT. Note that since this also recovers $ATT(a)$, it is equivalent--in terms of identification--to our proposed estimator.}
\end{remark}

\begin{remark}[Adding covariates for precision]
Researchers often include covariates to improve precision. If the covariates only affect precision, that is, they do not induce heterogeneity in the bias of conditional parallel trends, then the bias term in Theorem~\ref{theorem.1} cancels out across groups. In this case, the difference in group-specific DID estimators with covariates still identifies the target causal estimand. However, if the covariates both improve precision and introduce heterogeneity in the bias of parallel trends, then potentially they also introduce bias.
\end{remark}

\begin{remark}[Comparing regressions with and without covariates] 
Researchers often compare TDID regressions with and without covariates, but this comparison is difficult to interpret. Differences may arise due to both estimator bias and changes in the target estimand. Using the conventional approach, adding covariates can introduce bias unless covariate distributions are balanced across groups. Moreover, if the setting satisfies Assumption~\ref{A.treatment_ass}(ii), the target estimand itself changes: without covariates, the target is $ATT(a) - ATT(b)$; with covariates and reweighting (Theorem~\ref{theorem.2}), it becomes $ATT(a) - \mathbb{E}[CATT(b,X)\mid G_i = a, E_i = 2]$.\footnote{If covariate distributions are equivalent across groups, then the bias is zero, and $\mathbb{E}[CATT(b,X)\mid G_i = a, E_i = 2] = ATT(b)$. Hence, putting aside numerical considerations, there is no difference with and without controls.}
\end{remark}

\begin{remark}[Aggregated covariates] 
In settings where units are aggregates (e.g., states or counties) and groups refer to individuals within those units, covariates often do not vary by group. For example, \textcite{goldin2013shocking} compare married and non-married women but control for state-level covariates that are constant across groups.\footnote{In the replication data and code for \textcite{goldin2013shocking}, the covariates are not group-invariant. A similar issue arises in the replication for the TDID example in \textcite{cunningham2021causal}.} When covariates do not vary across groups, they offer no identifying variation in a TDID design, making their inclusion unnecessary from an identification perspective.
\end{remark}

\printbibliography

\appendix

\section{Details on Survey}\label{sec:appendix_survey}

The literature survey was conducted using Google Scholar in February 2025 to search for papers that used the TDID with controls empirical strategy. We searched for the term "triple difference" and "triple differences", both with and without a semicolon. We restricted to papers with above 100 citations, published between 2010 and 2025 in (number of surveyed papers in parentheses): American Economic Review including Insights and Papers \& Proceedings (33), Journal of Political Economy (9), Quarterly Journal of Economics (4), Econometrica (1), Review of Economic Studies (6) and Review of Economics and Statistics (13). The journals were chosen from the top journals in Google Scholar's classification based on h5-index for ''Business, Economics and Management'' journals (financial journals were excluded). 

Within each paper, we searched for the first mention of the string "triple", and considered the nearest empirical specification which the authors referred to as TDID. These specifications were manually classified as fully-interacted, differenced, or other. Specifications were classified as fully-interacted if they include a three-way interaction term. Furthermore, we manually coded whether the specification included controls variables, where double-interaction terms or fixed effects of double-interactions were not classified as controls. Similarly, variables that were part of the three-way interactions, or fixed effects thereof, were not classified as controls. If multiple versions of the same specification were reported in a table, and in at least one version controls were added, then we classify as including controls. Five papers in which the specification was not clear or there was no use of the TDID design, and yet did come up in the search, were excluded. After omitting, the number of papers included in the survey is 66.

\section{Proofs}\label{sec:appendix_proofs}

Before proving the theorems, we state and prove a useful lemma. 
\begin{lemma}\label{lemma.catt_diff}
Assume the conditions in Theorem \ref{theorem.1} hold. Then for all $x\in\mathcal{X}$, under Assumption \ref{A.treatment_ass}(ii)
$$CATT(a,x)-CATT(b,x)=D\left(a,x\right)-D\left(b,x\right)$$
and under Assumption \ref{A.treatment_ass}(i)
$$CATT(a,x)=D\left(A,x\right)-D\left(b,x\right)$$
\end{lemma}

\begin{proof}[Proof of Lemma \ref{lemma.catt_diff}]
First, consider the setup under Assumption \ref{A.treatment_ass}(ii). For $g\in\{a,b\}$ Assumption \ref{A.same_bias} implies
\begin{align*}
\mathbb{E}\big[Y_{2}(0)\mid G=g,  E_{i}&=2,X=x\big]  =\mathbb{E}\big[Y_{1}(0)\mid G=g,E=2,X=x\big]\\
 & +\mathbb{E}\big[Y_{2}(0)-Y_{1}(0)\mid G=g,  E=\infty,X=x\big] + \delta_{\mathrm{CPT}}\left(x\right).
\end{align*}
From Assumption \ref{A.no.ant.1}
\begin{align*}\mathbb{E}\big[Y_{2}(0)\mid G=g,  E&=2,X=x\big]  =\mathbb{E}\big[Y_{1}\mid G=g,E=2,X_{i}=x\big]\\
 & +\mathbb{E}\big[Y_{2}-Y_{1}\mid G=g,  E=\infty,X=x\big] + \delta_{\mathrm{CPT}}\left(x\right).
\end{align*}
From consistency 
$$\mathbb{E}\big[Y_{2}(1)\mid G=g,E=2,X=x\big]=\mathbb{E}\big[Y_{2}\mid G=g,E=2,X=x\big].$$
Substituting into $CATT\left(g,x\right)$ implies
\begin{equation}\label{eq:lemma.1}
CATT\left(g,x\right) = D\left(g,x\right)-\delta_{\mathrm{CPT}}\left(x\right).
\end{equation}
Differencing \eqref{eq:lemma.1} when setting $g=a$ from \eqref{eq:lemma.1} when setting $g=b$, the terms $\delta_{\mathrm{CPT}}\left(x\right)$ cancel out, and the statement of the lemma is achieved.

Next, consider the setup under Assumption \ref{A.treatment_ass}(i). Similar arguments based on Assumption \ref{A.no.ant.1}, Assumption \ref{A.same_bias} and consistency imply
\begin{equation}\label{eq:lemma.2}
D\left(b,x\right)=\delta_{\mathrm{CPT}}\left(x\right).
\end{equation}
The statement in the lemma is achieved by summing \eqref{eq:lemma.1} when setting $g=a$ and \eqref{eq:lemma.2}. \qed
\end{proof}
\vspace{10pt}
\begin{proof}[Proof of Theorem \ref{theorem.1}]
Consider the setup under Assumption \ref{A.treatment_ass}(ii). Taking expectations conditional on $G_i=g$ and $E_i=2$ over both sides of \eqref{eq:lemma.1}, the law of iterated expectation (LIE) and linearity of expectations imply
\begin{equation}\label{eq:thereom1.1}
ATT\left(g\right) =\mathbb{E}\left[D\left(g,X\right)\mid G=g,E=2\right] -\mathbb{E}\left[\delta_{\mathrm{CPT}}\left(X\right)\mid G=g,E=2\right]. 
\end{equation}
Differencing \eqref{eq:thereom1.1} when setting $g=a$ from \eqref{eq:thereom1.1} when setting $g=b$ we obtain the statement in the theorem. Similarly, the statement under Assumption \ref{A.treatment_ass}(i) is achieved by taking expectations conditional on $G_i=b$ and $E_i=2$ over \eqref{eq:lemma.2}, and differencing from \eqref{eq:thereom1.1} when setting $g=a$.\qed
\end{proof}
\vspace{10pt}
\begin{proof}[Proof of Theorem \ref{theorem.2}]
Consider the setup under Assumption \ref{A.treatment_ass}(ii). Setting $g=b$ in \eqref{eq:lemma.1}, and taking conditional expectations over $G_i=a$ and $E_i=2$ implies
\begin{align}\label{eq:thereom2.1}
\mathbb{E}\left[CATT\left(b,X\right)\mid G=a,E=2\right] & =\mathbb{E}\left[D\left(b,X_i\right)\mid G=a,E=2\right] \nonumber\\
 & -\mathbb{E}\left[\delta_{\mathrm{CPT}}\left(X\right)\mid G=a,E=2\right].
\end{align}
Differencing from  \eqref{eq:thereom1.1} when setting $g=a$ from \eqref{eq:thereom2.1} we obtain the statement in the theorem.
Next, consider the setup under Assumption \ref{A.treatment_ass}(i). Taking expectations conditional on $G_i=a$ and $E_i=2$ over \eqref{eq:lemma.2}, and differencing from \eqref{eq:thereom1.1} when setting $g=a$, we obtain the statement in the theorem.\qed
\end{proof}
\vspace{10pt}
\noindent

\begin{proof}[Proof of Lemma \ref{lemma.wdr}]
Let 
$\Psi(e) = \mathbb{E}\left[\mathbb{E}\left[Y_{2}-Y_{1}\mid G=b,E=e,X=x\right]\mid G=a,E=2\right].$
The expression on the left-hand-side in Lemma \ref{lemma.wdr} can be re-written $\Psi(2)-\Psi(\infty)$. Below we establish general statements on $\Psi(e)$, which can then be applied to this difference. First, we obtain 
\begin{align*}
 \Psi(e)& = \mathbb{E}\left[\mathbb{E}\left[Y_{2}-Y_{1}\mid G=b,E=e,X=x\right]\mid G=a,E=2\right]\\
 & =\mathbb{E}\left[\frac{1_{\left\{ G=a,E=2\right\} }}{\mathbb{E}\left[1_{\left\{ G=a,E=2\right\} }\right]}\mathbb{E}\left[Y_{2}-Y_{1}\mid G=b,E=e,X=x\right]\right] \\
 & = \mathbb{E}\left[w_{\mathrm{T}}\left(a,2\right)m\left(b,e,X\right)\right].
\end{align*}
This shows the equality for $\mathbb{E}[\psi_{WOR}]$. Next, 
\begin{align}\label{eq:lemma2.1}
 \Psi(e)& = \mathbb{E}\left[\frac{1_{\left\{ G=a,E=2\right\} }}{\mathbb{E}\left[1_{\left\{ G=a,E=2\right\} }\right]}\mathbb{E}\left[Y_{2}-Y_{1}\mid G=b,E=e,X=x\right]\right]\nonumber\\
 & =\mathbb{E}\left[\frac{\mathbb{E}\left[1_{\left\{ G=a,E=2\right\} }\mid X=x\right]}{\mathbb{E}\left[1_{\left\{ G=a,E=2\right\} }\right]}\mathbb{E}\left[Y_{2}-Y_{1}\mid G=b,E=e,X=x\right]\right]\nonumber\\
 & =\mathbb{E}\left[\frac{1_{\left\{ G=b,E=e\right\} }}{\mathbb{E}\left[1_{\left\{ G=a,E=2\right\} }\right]}\frac{\mathbb{E}\left[1_{\left\{ G=a,E=2\right\} }\mid X=x\right]}{\mathbb{E}\left[1_{\left\{ G=b,E=e\right\} }\mid X=x\right]}\mathbb{E}\left[Y_{2}-Y_{1}\mid X=x\right]\right]\nonumber\\
 & =\mathbb{E}\left[\frac{1_{\left\{ G=b,E=e\right\} }}{\mathbb{E}\left[1_{\left\{ G=a,E=2\right\} }\right]}\frac{p\left(a,2,X\right)}{p\left(b,e,X\right)}\mathbb{E}\left[Y_{2}-Y_{1}\mid X=x\right]\right]\nonumber\\
 & =\mathbb{E}\left[\frac{1_{\left\{ G=b,E=e\right\} }}{\mathbb{E}\left[1_{\left\{ G=a,E=2\right\} }\right]}\frac{p\left(a,2,X\right)}{p\left(b,e,X\right)}\left(Y_{2}-Y_{1}\right)\right] \nonumber\\
 & =\mathbb{E}\left[w_{\mathrm{C}}\left(a,b,2,e,X\right)\left(Y_{2}-Y_{1}\right)\right].
\end{align}
This shows the equality for $\mathbb{E}[\psi_{WIPW}]$. The above arguments also imply 
$$\mathbb{E}\left[\left(w_{\mathrm{T}}\left(a,2\right)-w_{\mathrm{C}}\left(a,b,2,e,X\right)\right)m\left(b,e,X\right)\right]=0.$$
Hence, adding zero to \eqref{eq:lemma2.1} we obtain 
\begin{align}\label{eq:lemma2.3}
 \Psi(e)& =\mathbb{E}\left[w_{\mathrm{C}}\left(a,b,2,e,X\right)\left(Y_{2}-Y_{1}\right)\right]\nonumber\\
 & +\mathbb{E}\left[\left(w_{\mathrm{T}}\left(a,2\right)-w_{\mathrm{C}}\left(a,b,2,e,X\right)\right)m\left(b,e,X\right)\right].
\end{align}
This shows the equality for $\mathbb{E}[\psi_{WDR}]$. \qed
\end{proof}

\section{An Example}\label{sec:example}

To illustrate the results in Theorems 1 and 2 in the main text, we present a toy analytical example. Assume a  setup with panel data with two time periods $t\in\{1,2\}$, where units randomly belong to group A or B. Units draw a time-invariant covariate $X_i$ according to $X_i\sim N(\mu_g,1)$, with group-specific means $\mu_a = 1$ and $\mu_b=3$.
Eligibility in period 2 is assigned randomly to units from both groups. The treatment assignment follows Assumption 1(ii), as in, eligible units are treated in both groups. Untreated potential outcomes evolve according to 
\begin{equation*}
    Y_{i,1}(0)=X_{i}+\varepsilon_{i,1}, \quad\quad Y_{i,2}(0)=X_{i}+W_{i,2}X_{i}+\varepsilon_{i,2}, \quad\quad \varepsilon_{i,1},\varepsilon_{i,2}\sim N(0,1).
\end{equation*}

\subsection{Group invariant biases in conditional parallel trends} 

We begin by showing that the above setup satisfies Assumption 3. Recall that conditional parallel trends for group $g$ is equal to 
\begin{align*}
\delta_{\mathrm{CPT}}\left(g,x\right) & =\mathbb{E}\left[Y_{2}\left(0\right)-Y_{1}\left(0\right)\mid G=g,E=2,X=x\right]\\
 & -\mathbb{E}\left[Y_{2}\left(0\right)-Y_{1}\left(0\right)\mid G=g,E=\infty,X=x\right].
\end{align*}
Substituting the definitions in the setup, we obtain 
\begin{align*}
\delta_{\mathrm{CPT}}\left(g,x\right) & =\mathbb{E}\left[\left(X+W_{2}X+\varepsilon_{2}\right)-\left(X+\varepsilon_{1}\right)\mid G=g,E=2,X=x\right]\\
 & -\mathbb{E}\left[\left(X+W_{2}X+\varepsilon_{2}\right)-\left(X+\varepsilon_{1}\right)\mid G=g,E=\infty,X=x\right],
\end{align*}
which simplifies to $\delta_{\mathrm{CPT}}\left(g,x\right)=x$, since non-eligible units have $W_{i,2}=0$ and the errors are mean zero. Hence, conditional parallel trends are violated, but the bias is group-invariant, satisfying Assumption 3.

\subsection{Case 1. Constant treatment effects}

Assume that $Y_{i,t}(1)=Y_{i,t}(0)+\beta(G_i)$ where $\beta(a)=4$ and $\beta(b)=1$. Also, assume observed outcomes are given by $Y_{i,t}=W_{i,t}Y_{i,t}(1)+(1-W_{i,t})Y_{i,t}(0)$. 

In the above setup, we can analytically calculate the causal and descriptive estimands of interest. First, since the treatment effect is not heterogeneous by $X$, $CATT\left(g,x\right)=ATT(g)=\beta(g)$. This implies that $ATT\left(a\right)=4$ and $ATT\left(b\right)=1$, and hence the difference in ATTs is equal to three. 

Next, turn to the group-specific conditional DID. Substituting observed and potential outcomes, and simplifying, we obtain 
$D_{\mathrm{CDID}}\left(g,x\right)=\beta\left(g\right)+x$. 
Taking expectations over $X$ conditional on group and eligibility we have 
\begin{align*}
\mathbb{E}\left[D_{\mathrm{CDID}}\left(a,x\right)\mid G=a,E=2\right] & =\beta\left(a\right)+\mathbb{E}\left[X\mid G=a,E=2\right]=4+1=5,\\
\mathbb{E}\left[D_{\mathrm{CDID}}\left(b,x\right)\mid G=b,E=2\right] & =\beta\left(b\right)+\mathbb{E}\left[X\mid G=b,E=2\right]=1+3=4.
\end{align*}

Differencing the expectation for A from the expectation from B we obtain  one. Recall that the difference in ATT of A from the ATT of B is three. This discrepancy arises from the aggregate bias in conditional parallel trends, equal to one for group A and three for group B, resulting in a difference of minus two. This shows that (i) the DID is biased for the ATT of each group, due to the violation of conditional parallel trends, and (ii) the difference in DID between groups is biased for the difference in ATTs between groups, due to the differences in the covariate distribution of the groups (Theorem 1). 

Finally, the expectation of the conditional DID for group B, conditional on group A, is 
$$\mathbb{E}\left[D_{\mathrm{CDID}}\left(b,x\right)\mid G=a,E=2\right]=\beta\left(b\right)+\mathbb{E}\left[X\mid G=a,E=2\right]=1+1=2$$
and hence
$$\mathbb{E}\left[D_{\mathrm{CDID}}\left(a,X\right)\mid G=a,E=2\right]-\mathbb{E}\left[D_{\mathrm{CDID}}\left(b,X\right)\mid G=a,E=2\right]=5-2=3.$$
Under treatment effect homogeneity, this identifies the difference in ATTs (Theorem 2). In the next subsection, we turn to a case of heterogeneous treatment effects. 

\subsection{Case 2. Heterogeneous treatment effects}

Now, assume that $Y_{i,t}(1)=Y_{i,t}(0)+\beta(G_i,X_i)$ where $\beta(a,x)=4x$ and $\beta(b,x)=x$.
Following the steps in the previous case, we obtain $CATT\left(g,x\right)=\beta(g,x)$. Taking expectations over the CATT by group, $ATT(a)=4$ and $ATT(b)=3$, and the difference in ATTs is equal to one. 
Next, the conditional DID of each group in this case is equal to $D_{\mathrm{CDID}}\left(g,x\right)=\beta\left(g,x\right)+x$. 
Taking expectations 
\begin{align*}
\mathbb{E}\left[D_{\mathrm{CDID}}\left(a,x\right)\mid G=a,E=2\right] & =\mathbb{E}\left[5X\mid G=a,E=2\right]=5,\\
\mathbb{E}\left[D_{\mathrm{CDID}}\left(b,x\right)\mid G=b,E=2\right] & =\mathbb{E}\left[2X\mid G=b,E=2\right]=6.
\end{align*}
Hence, the difference in expected values of conditional DID by group is minus one. As expected, this is not equal to the difference in ATTs, due to the difference in the distribution of $X$. Finally, 
$$\mathbb{E}\left[D_{\mathrm{CDID}}\left(b,x\right)\mid G=a,E=2\right]=\mathbb{E}\left[2X\mid G=a,E=2\right]=2$$
and hence
$$\mathbb{E}\left[D_{\mathrm{CDID}}\left(a,X\right)\mid G=a,E=2\right]-\mathbb{E}\left[D_{\mathrm{CDID}}\left(b,X\right)\mid G=a,E=2\right]=5-2=3.$$
This is not equal to the difference in ATTs, which, recall, equals one. However, we can calculate the identifiable estimand. Noting that $$\mathbb{E}\left[CATT(b,X)\mid G=a,E=2\right]=\mathbb{E}\left[X_i\mid G=a,E=2\right]=1$$
we obtain
$\mathbb{E}\left[CATT(a,X)-CATT(b,X)\mid G=a,E=2\right]=4-1=3$.
This confirms that the expected difference in conditional DIDs evaluated under group A’s covariate distribution, is equal to the expected difference in CATTs, also evaluated under the covariate distribution of group A (Theorem 2).

\section{Simulation Study}\label{sec:simulation}

We illustrate the performance of the proposed estimators using a Monte-Carlo simulation study. Assume the data generating process (DGP) is equivalent to the setup in the example (Appendix \ref{sec:example} Case 2). From the DGP we draw $2,000$ samples of size $n=2,000$. In each sample we estimate (i) the DR estimators of both groups separately, and (ii) the WDR estimator of group B. Then we construct the (wrong) difference of the DR estimators of groups A and B (Theorem 1), and the (correct) difference of the DR estimator of group A from the WDR estimator of group B (Theorem 2). Propensity scores were modeled using logistic regressions and outcome regressions were modeled using ordinary least squares (OLS), without cross-fitting. 

Figure \ref{fig:1} presents the empirical distribution of the estimators. The results are an empirical equivalent of the analytical results we derived in Appendix \ref{sec:example}: the difference between DR estimators has a mean of minus one, while the difference between the DR estimator of group A and the WDR estimator of group B has a mean of three. Recall that the difference in ATTs between the groups was one, and the difference in the ATT of group A and the expected value of the CATT of group B according to the covariate distribution of group A was three. 

\begin{figure}[t]
    \centering
    \caption{Simulation Results}
    \label{fig:1}
    \includegraphics[width=\textwidth]{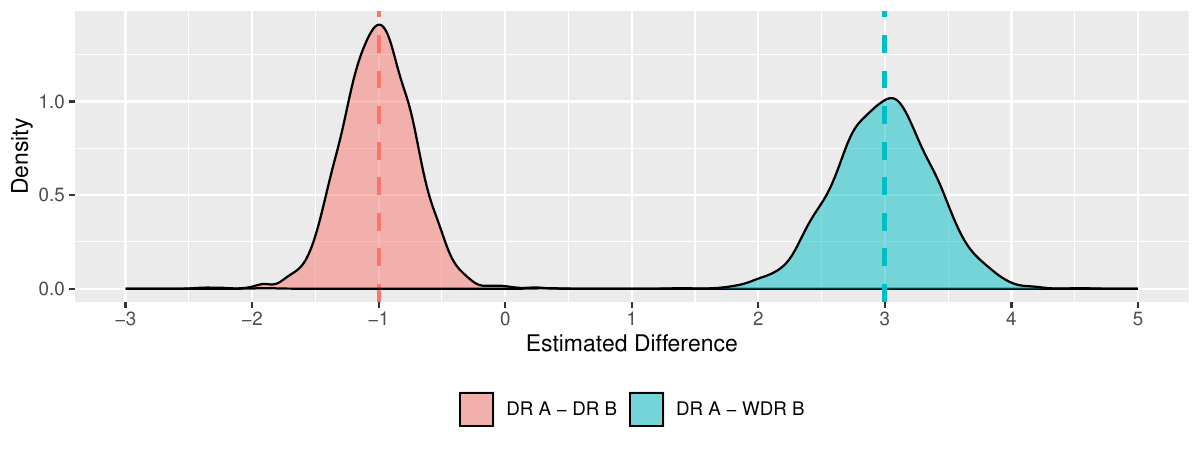}
    \caption*{\footnotesize \textit{Notes:}
    The figure presents empirical distributions of estimators from a Monte-Carlo simulation study. The DGP is based on the example in Appendix \ref{sec:example}.  "DR A - DR B" represents an estimator which takes the difference between the DR estimator of group A and the DR estimator of group B. "DR A - WDR B" represents the proposed estimator in Section III in the main text. The distribution is based on 2,000 samples of sample size 2,000.
    }
\end{figure}

\end{document}